\documentclass[runningheads]{llncs}
\usepackage[utf8]{inputenc}
\usepackage{amsmath, amssymb}
\usepackage{graphicx}
\usepackage{xspace}
\usepackage{xcolor}
\usepackage{extarrows}
\usepackage{marginnote}
\usepackage{xargs}
\usepackage{tikz}
\usetikzlibrary{positioning, calc, shapes, automata, shapes.multipart, arrows, arrows.meta}
\usepackage{hyperref}
\usepackage{url}
\usepackage{paralist} 
\usepackage{caption} 
\usepackage{subcaption} 

\allowdisplaybreaks

\newcommand*{\tpkt}{\rlap{$\;$.}}

\newcommand{\defsym}{\mathrel{:=}}

\newcommand{\m}[1]{\mathsf{#1}}
\newcommand{\mc}[1]{\mathcal{#1}}
\newcommand{\II}{\mc{I}}
\newcommand{\RR}{\mc R}

\newcommand{\OO}{\mc O}
\newcommand{\DD}{\mc D}
\renewcommand{\SS}{\mc S}
\renewcommand{\AA}{\mc A}                   
\newcommand{\TT}{\mc T}
\newcommand{\FF}{\mc F}
\newcommand{\MM}{\mc M}
\newcommand{\VV}{\mc V}
\newcommand{\UU}{\mc U}
\newcommand{\JJ}{\mc{J}}
\newcommand{\FFt}{\FF_{T}}
\newcommand{\FFtD}{\FF_{\DD}}
\newcommand{\FFl}{\FF_{L}}
\newcommand{\Var}{\mc Var}
\newcommand{\Varvec}{\vec{\mc Var}}

\newcommand{\Val}{\mc{V}\m{al}}
\newcommand{\Pos}{\mc{P}\m{os}}

\newcommand{\PosD}{\Pos_{\FFtD}}

\newcommand{\lterms}{\TT(\FFl, \VV)}

\newcommand{\judge}[2]{{\vdash\,}{#1}\colon{#2}}
\newcommand\cprobx[1][(t_0,\varphi_0)]{(#1, \DD, \RR)}
\newcommand\cprob{(s_0, \DD, \RR)}
\newcommand{\llangle}{\langle\!\langle}
\newcommand{\rrangle}{\rangle\!\rangle}
\newcommand{\init}{\m{init}}

\newcommand{\seq}[2][n]{{#2_1},\dots,{#2_{#1}}}
\renewcommand{\root}{\m{root}}

\newcommand\interp[2][\JJ]{[#2]_{#1}}
\newcommand{\bool}{\mathsf{{bool}}}
\newcommand{\ints}{\mathsf{{int}}}
\newcommand{\lists}{\mathsf{{list}}}
\newcommand{\true}{\mathsf{{true}}}
\newcommand{\false}{\mathsf{{false}}}
\newcommand{\equals}{\mathrel{\approx}}
\newcommand{\calc}{\m{calc}}
\newcommand{\tocalc}{\to_\calc}

\newcommand{\tocmr}[1][\RR]{\to_\m{\RR/calc}}
\newcommand{\cgt}[1]{>_{[#1]}}
\newcommand{\cge}[1]{\geqslant_{[#1]}}
\renewcommand{\vec}[1]{\overline{#1}}
\newcommand{\tvar}[2]{(#1, #2)}
\newcommand{\bv}[1]{\textsf{\#x{#1}}}

\newcommand{\rc}{\m{rc}}
\renewcommand{\dh}{\m{dh}}
\newcommand{\UB}{\textup{UB}}
\newcommand{\TV}{\textup{EV}}
\newcommand{\REG}{G}
\newcommand{\TVG}{G_\TV}
\newcommand{\RA}{T}
\newcommand{\SA}{S}
\newcommand{\CA}{C}
\newcommand{\pre}{\m{pre}}
\newcommand{\DT}{\m{DT}}

\newcommand{\mergeSort}{\mathsf{m}}
\newcommand{\cons}{\mathrel{::}}

\newcommand{\Nats}{\mathbb N}
\newcommand{\Ints}{\mathbb Z}

\renewcommand{\leq}{\leqslant} 
\renewcommand{\geq}{\geqslant}

\def\test#1#2#3{\setbox0=\hbox{$\vphantom{#1}^{#2}_{#3}$}%
                \dimen0=\wd0%
                \setbox1=\hbox{$\scriptstyle #2$}%
                \advance\dimen0-\wd1%
                \setbox1=\hbox{\hskip\dimen0\copy1}%
                \dimen0=\wd0%
                \setbox2=\hbox{$\scriptstyle #3$}%
                \advance\dimen0-\wd2%
                \setbox2=\hbox{\hskip\dimen0\copy2}%
                {\vphantom{#1}^{\box1}_{\box2}}{#1}
}

\newcommand{\subsumeseq}{\mathrel{\makebox[0pt]{\makebox[9pt][r]%
{\raise 0.9pt \hbox{$\cdot$}}}{\geq}}}
\newcommand{\ito}{\overset{\mathsf{i}}{\to}}

\newlength{\myline}
\newcommandx*{\triplearrow}[4][1=4, 2=4]{
  \draw[line width=\myline,shorten <=#1\myline,shorten >=#2\myline, double distance=3\myline,#3] #4;
  \draw[line width=\myline,shorten <=#1\myline,shorten >=#2\myline,#3] #4;
}

\makeatletter
\newcommand{\xRightarrow}[2][]{\ext@arrow 0359\Rightarrowfill@{#1}{#2}}
\makeatother

\newcommand\TTTT{%
 \textsf{T\kern-0.15em\raisebox{-0.55ex}T\kern-0.15emT\kern-0.15em\raisebox{-0.55ex}2}%
}
\newcommand\aprove{\textsf{AProVE}\xspace}
\newcommand\koat{\textsf{KoAT}\xspace}
\newcommand\pubs{\textsf{PUBS}\xspace}
\newcommand\cofloco{\textsf{CoFloCo}\xspace}
\newcommand\tctlctrs{\TCT-\textsf{LCTRS}\xspace}
\newcommand{\TCT}{\textsf{$\textsf{T}\!\protect\raisebox{-1mm}{C}\!\textsf{T}\!$}\xspace}

\newcounter{commentcount}

\newcommand{\lemref}[1]{Lem.~\ref{lem:#1}}
\newcommand{\lemsref}[2]{Lems.~\ref{lem:#1} and \ref{lem:#2}}
\newcommand{\defref}[1]{Def.~\ref{def:#1}}

\newcommand{\secref}[1]{Sect.~\ref{sec:#1}}

\newcommand{\exaref}[1]{Ex.~\ref{exa:#1}}
\newcommand{\figref}[1]{Fig.~\ref{fig:#1}}
\newcommand{\tabref}[1]{Tab.~\ref{tab:#1}}
\renewcommand{\eqref}[1]{\ref{eq:#1}}
\newcommand{\etal}{\textsl{et al.}\xspace}

\newtheorem{numberedlemma}{Lemma}

\begin{document}

\title{Runtime Complexity Analysis of Logically Constrained Rewriting}
\titlerunning{Runtime Complexity Analysis of Logically Constrained Rewriting}
\author{Sarah Winkler 
and Georg Moser
}
\authorrunning{S. Winkler and G. Moser}
\institute{Free University of Bolzano, Italy, and University of Innsbruck, Austria\\
\email{sarwinkler@unibz.it,georg.moser@uibk.ac.at}}

\maketitle

\begin{abstract}
Logically constrained rewrite systems (LCTRSs) are a versatile and efficient 
rewriting formalism that can be used to model programs from various 
programming paradigms, as well as simplification systems in compilers and 
SMT solvers.
In this paper, we investigate techniques to analyse the worst-case runtime complexity of LCTRSs. For that, we exploit synergies between previously developed decomposition techniques for standard term rewriting by Avanzini \etal in conjunction with alternating time and size bound approximations for integer programs by Brockschmidt \etal and adapt these techniques suitably to LCTRSs.
Furthermore, we provide novel modularization techniques to exploit loop bounds from recurrence equations which yield sublinear bounds.
We have implemented the method in \TCT to test the viability of our method. 
\end{abstract}

\section{Introduction}
\label{sec:intro}

Rewriting with constraints over background theories is a highly 
versatile model of computation and tool for analysis.
While user-defined data types are modelled by free function symbols, arbitrary 
decidable theories can be incorporated, such as integer or bit-vector 
arithmetic, lists, or array theory. Constraints over these theories can be effectively handled by SMT solvers.
Different rewrite formalisms capture this idea~\cite{maude,FNSKS08,FGPSF09}. 
Here we use the recent notion of \emph{logically constrained term rewrite 
systems} (\emph{LCTRSs} for short), due to Kop et al.~\cite{KN13,ctrl,FKN17,CL18}.

LCTRSs can abstract programs in a variety of paradigms, comprising 
imperative, functional, and logic languages. They also subsume 
integer transition systems (ITSs), which constitute a frequently used
program abstraction~\cite{FGPSF09,FKS11,BEFFG16} but do---in contrast to 
LCTRSs---not support (non-tail) recursion.
On the other hand, LCTRSs can also model simplification routines for 
expressions, which are crucial procedures in compilers or SMT solvers.
For all of these application areas, LCTRSs offer a \emph{uniform} toolset 
to analyse \emph{termination} (or non-ter\-mination)~\cite{K13,NW18}, \emph{reachability}~\cite{CL18}, 
\emph{uniqueness}~\cite{WM18}, or \emph{program equivalence}~\cite{FKN17}.

However, techniques for resource analysis of LCTRSs are so far lacking. 
This is
despite the fact that in their application domains (program analysis,
simplification systems), execution time is crucial. As a remedy, this paper
investigates methods to analyse (worst-case) innermost runtime
complexity of logically constrained rewrite systems.
To this end, we unify and generalise the complexity framework for standard rewriting by 
Avanzini and Moser~\cite{AM16} with the approach by Brockschmidt \etal to
alternate time and size bound analysis for ITSs~\cite{BEFFG16}, and moreover 
propose processors for modularisation and sublinear bounds.

\paragraph{Contributions.} We present a novel resource analysis framework for logically constrained rewrite systems (\secref{framework})
coached in the modular processor framework of~\TCT~\cite{AMS16}. Precisely, 
\begin{enumerate}
\item we present the first fully-automated runtime complexity analysis of LTCRSs;
\item we unify the complexity framework for standard (innermost) rewriting by
Avanzini and Moser~\cite{AM16} and the alternating time and size bound approximations
for ITSs by Brockschmidt \etal~\cite{BEFFG16},
\item generalising this, we introduce a novel modularisation processor, the \emph{splitting processor}; 
\item we present a novel processor, dubbed \emph{recurrence processor} to derive sublinear bounds based on recurrences as described by the Master Theorem;
\item we illustrate the viability of our method by providing a prototype implementation as a dedicated module \texttt{tct-lctrs} in~\TCT, and evaluate it on ITS benchmarks.
\end{enumerate}

In the remainder of the section, we highlight potential application areas of 
LCTRSs to emphasise their versatility.
In the next section (\secref{mergesort}) we give a high-level account of our technical achievements, providing a step-by-step explanation
how the runtime complexity of a natural representation of \textsf{mergesort} can be optimally analysed
in our framework. In this section, we also discuss to what extent our results can be applied to the below given examples.
In~\secref{basics} we summarise the foundations of LCTRSs, while in~\secref{framework} we detail the complexity framework
used. Processors carried over from the ITS setting are presented in \secref{procs}, and the novel processors are introduced in~\secref{newprocs}. Implementational choices and experimental results are summarized in~\secref{implementation}.
Finally, in \secref{conclusion} we conclude.
Due to space restrictions, some proofs were moved to an appendix.

\paragraph{Logically Constrained Rewrite Systems.}
We emphasise motivational examples from three different domains, focusing 
on imperative and logic programs, as well as compiler optimisations.

\begin{example}
\label{exa:mergesort}
The following recursive ITS $\RR_1$, due to Albert et al.~\cite{AAGP08},
corresponds to an imperative \textsf{mergesort} implementation after computing
loop summaries. It is naturally coached into the LCTRS framework, with the theory of 
integers as background theory.
\begin{footnotesize}
\begin{xalignat*}{4}
(1)\!\!\!\! &&
  \init(x,y,z) &\to \mergeSort(x,y,z) &
(2)\!\!\!\! &&
  \mergeSort_3(x,y,z) &\to \m{merge}(y,z,z) \\
(3)\!\!\!\! &&
  \mergeSort_1(x,y,z) &\to \mergeSort(y,y,z) &
(4)\!\!\!\! &&
  \m{merge}(x,y,z) &\to \m{merge}(x - 1,y,z)\ [x \geq 1 \wedge y \geq 1] \\
(5)\!\!\!\! &&
  \mergeSort_0(x,y,z) &\to \m{split}(x,y,z) &
(6)\!\!\!\! &&
  \m{split}(x,y,z) &\to \m{split}(x - 2,y,z)\ [x \geq 2] \\
(7)\!\!\!\! &&
  \mergeSort_2(x,y,z) &\to \mergeSort(z,y,z) &
(8)\!\!\!\! &&
  \m{merge}(x,y,z) &\to \m{merge}(x,y - 1,z)\ [x \geq 1 \wedge y \geq 1] \\
(9)\!\!\!\! &&
  \mergeSort(x,y,z) &\to \rlap{$\langle \mergeSort_0(x,u,v),\mergeSort_1(x,u,v),\mergeSort_2(x,u,v),\mergeSort_3(x,u,v)\rangle$}\\&&& 
  \rlap{$[x \geq 2 \wedge u \geq 0 \wedge v \geq 0 \wedge x + 1 \geq 2u \wedge 2u \geq x \wedge x \geq 2v \wedge 2v + 1 \geq x]$}
\end{xalignat*}
\end{footnotesize}
Here a rule of the form $\ell \to r\ [c]$ means that an instance of $\ell$ is replaced by the respective instance of $r$ provided that the instance of $c$ is satisfied.
\end{example}
\noindent
Similarly, (constraint) logic programs can be nicely suited to LCTRSs.

\begin{example}
\label{exa:maxlength}

Consider the following simple Prolog program from the benchmarks collected by
Mesnard and Neumerkl~\cite{MN01}.%
\begin{footnotesize}
\texttt{
\begin{center}
\begin{tabular}{r@{\ }l@{\quad}@{\qquad}l}
max\_length(Ls,M,Len) &:- max1(Ls,0,M), len(Ls,Len). \\
len([H|T],L) & :- len(T,LT), L is LT + 1. &
len([],0). \\
max1([H|T],N,M) &:- H <= N, max1(T,N,M). &
max1([],M,M). \\
max1([H|T],N,M) &:- H > N, max1(T,H,M).
\end{tabular}
\end{center}}
\end{footnotesize}
\noindent
Assuming an instantiated list $\mathtt{Ls}$, \texttt{max\_length(Ls,M,Len)}
is deterministic and returns the maximal list entry and 
the length of the list. This function becomes representable as
the following LCTRS $\RR_2$ over the theory of integers and lists:
\begin{xalignat*}{2}
\m{max\_length}(ls,m,l) &\to \langle\m{max}(ls,\m 0,m), \m{len}(ls,l)\rangle\\
\m{len}(xs,l) & \to \m{len}(t,l-\m 1)~ [xs \approx h\cons t] &
\m{len}([],\m 0) &\to \langle\rangle \\
\m{max}(xs,n,m) &\to \m{max}(t,n,m)~[h \leq n\wedge xs \approx h\cons t]&
\m{max}([],m,m)&\to \langle\rangle \\
\m{max}(xs,n,m) &\to \m{max}(t,h,m) ~[h > n\wedge xs \approx h\cons t]
\end{xalignat*}
Here, $\cons$ denotes the cons operator and
$\langle \cdot, \cdot\rangle$, $\langle\rangle$ are additional constructor
symbols to collect the recursive calls of a rule.
Conceptually LCTRSs appear as a good fit to express \emph{constraint}
logic programs as well, making use of the
fact that constraints are natively supported.
\end{example}
\noindent
In order to emphasise that LCTRSs are not confined to 
static program analysis, we present a
final example which is concerned with program optimisation.
\begin{example}
\label{exa:alive}  
The Instcombine pass in the LLVM compilation suite
performs \emph{peephole optimisations} to simplify expressions
in the intermediate representation. The current optimisation set 
contains over 1000 simplification rules to e.g.\ replace multiplications
by shifts or perform bitwidth changes.
About 500 of them have recently been translated into the domain-specific
language \textsf{Alive}~\cite{LMNR18}, and  subsequently 
into LCTRSs~\cite{WM18}, resulting in rules of the following shape:
\begin{align*}
\m{add}(x,x) &\to \m{shift\_left}(x,\bv{1})\\
\m{add}(\m{add}(\m{xor}(\m{or}(x,a),y),\bv{1}),w)
&\to \m{sub}(w,\m{and}(x,b))~[ a \approx {\sim}b]\\
\m{add}(\m{xor}(x,a),z)
             &\to \m{sub}(a+z,x)~[\m{isPowerOf2}(a + \bv{1}) \land \dots]
               \tpkt
\end{align*}
These rules are expressed over the background theory of bit-vectors.
Naturally, as a compiler pass this simplification suite is a performance-%
critical routine, hence an automated complexity analysis is of great 
interest.
\end{example}

\section{Step by Step to an Optimal Bound}

\label{sec:mergesort}

Consider the rewrite system $\RR_1$ from \exaref{mergesort}, and
a rewrite sequence starting with an instance of $\init(x_0,y_0,z_0)$.
Below we sketch the steps to obtain an upper bound on the runtime complexity of $\RR_1$, expressed in $|x_0|$, $|y_0|$, and
$|z_0|$, where $|\cdot|$ denotes the absolute value.

An automated runtime complexity analysis of \textsf{mergesort} is notoriously difficult:
For this example, \cofloco~\cite{Montoya17,FM16} can only derive a quadratic bound, 
while \koat~\cite{BEFFG16} (as well as \aprove~\cite{Giesl:2017}) even proposes an exponential bound. \pubs~\cite{AAGP08} can produce an $\mathcal O(n\cdot \log(n))$ bound, using a special \emph{level-counting} feature, which however negatively affects its overall success rate.
Due to the work presented in this paper, our complexity analyser \TCT can automatically
prove the optimal $\mathcal O(n\cdot \log(n))$ upper bound. This is obtained by the following
recipe.
\begin{enumerate}
\item
We first compute \emph{dependency tuples} of all rules to focus the 
analysis on recursive calls (see \defref{dt}).
Then a \emph{dependency graph} approximation is computed
to estimate computation paths, where the numbers refer to the 
respective dependency tuples of rules in \exaref{mergesort}:
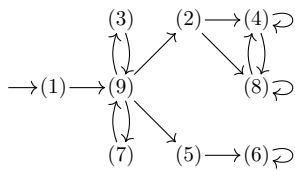
\begin{center}
\begin{tikzpicture}[yscale=.8, every loop/.style={min distance=3mm,in=-20,out=20,looseness=5}, node distance=9mm] 
\tikzstyle{to}=[->]
\tikzstyle{dt}=[inner sep=1pt, scale=.8]
\node[dt] (1) {(1)};
\node[right of=1, dt] (9) {(9)};
\node[above of=9, dt] (3) {(3)};
\node[below of=9, dt] (7) {(7)};
\node[right of=3, dt] (2) {(2)};
\node[right of=7, dt] (5) {(5)};
\node[right of=5, dt] (6) {(6)};
\node[right of=2, dt] (4) {(4)};
\node[below of=4, dt] (8) {(8)};
\draw[to] ($(1) + (-.6,0)$) -- (1);
\draw[to] (1) -- (9);
\draw[to] (9) to[bend left=15] (3);
\draw[to] (3) to[bend left=15] (9);
\draw[to] (9) to[bend left=15] (7);
\draw[to] (7) to[bend left=15] (9);
\draw[to] (9) -- (2);
\draw[to] (9) -- (5);
\draw[to] (2) -- (4);
\draw[to] (2) -- (8);
\draw[to] (8) to[bend left=15] (4);
\draw[to] (4) to[bend left=15] (8);
\draw[to] (8) to[loop right] (8);
\draw[to] (4) to[loop right] (4);
\draw[to] (5) -- (6);
\draw[to] (6) to[loop right] (6);
\end{tikzpicture}
\end{center}
\item
Next, we derive bounds on the \emph{size of variables} in left hand sides of rules, in
terms of the sizes of the variables in the initial term $\init(x_0,y_0,z_0)$.
For example, it is easy to check that for rule (9), $|x|$, $|y|$, and $|z|$ are bounded by $|x_0|$, $|y_0|$, and $|z_0|$, respectively, and all variables in other rules are bounded by $|x_0|$. This is established by the \emph{size bounds processor} (\lemref{size bounds}).
Formally, we adapt techniques developed for ITSs for that purpose~\cite{BEFFG16}.
\item 
We first derive time bounds for the SCCs $\{2,4,8\}$ and $\{6\}$ separately (\lemref{timebounds}).
Thus, using the size bounds from above and suitable \emph{interpretations}~\cite{AM16} (also called polynomial ranking functions~\cite{BEFFG16}) for LCTRSs, one can derive linear runtime bounds $2|x_0|+1$ and $|x_0|$ for these subproblems, respectively.
\item 
In order to analyse the SCC $\{3,7,9\}$, we first apply \emph{chaining}
to combine rule (9) with (3) and (7), respectively (eliminating symbols $\m m_1$ and $\m m_2$).
\item
  With respect to the modified rule (9) and the derived subproblem bounds, we exploit the \emph{loop processor} (\lemref{loops1}) to observe that its runtime can thus be overestimated by the following \emph{recurrence}
  equations.
\begin{align}
\label{eq:mergesort recursion}
  f(|x|,|y|,|z|) &= 2\cdot f(|x|/2,|x|/2,|x|/2) + 3|x|+1 &
  f(1, |y|, |z|) &= 0\quad
\end{align}
Solving the recurrences by the Master Theorem, implies an
overall runtime complexity of $\mathcal O(|x_0|\cdot \log(|x_0|))$ for $\RR_1$, as $|x|$ in rule (9) is bound by $|x_0|$.
\end{enumerate}

Wrt.\ $\RR_2$ from \exaref{maxlength}, we can fully automatically infer an
(asymptotic) optimal linear bound on the runtime complexity for the given instantiation.
Here, we take an instance of $\m{max\_length}(xs,z,l)$ as initial term.
As for comparison, note that the corresponding logic program cannot be handled by a dedicated variant of 
\aprove~\cite{GSSEF:2012} geared towards runtime complexity analysis of logic programs.
Only termination can be shown by the most recent version of \aprove~\cite{Giesl:2017}.
A priori, our approach is restricted to logic programs with 
instantiation patterns that ensure determinism and avoid failure, but in
the conclusion we discuss how to overcome this limitation.

Finally, \exaref{alive} cannot yet be handled, as a successful analysis requires the extension of the
proposed framework to \emph{(innermost) derivational complexity} {(i.e., 
the setting of arbitrary starting terms that may contain nested defined symbols).
This is subject to future work. However, we conceive
the work established in this paper as a solid first step towards the automated analysis of such systems.

\section{Logically Constrained Term Rewriting}
\label{sec:basics}

We assume familiarity with term rewriting~\cite{BN98,TeReSe},
but briefly recapitulate the notion of logically constrained
rewriting~\cite{KN13,FKN17} that our approach is based on.
We consider an infinite, sorted set of variables $\VV$ and a sorted signature
$\FF = \FFt \uplus \FFl$ such that $\TT(\FF,\VV)$ denotes the set of terms over 
this disjoint signature. Symbols in $\FFt$ are called \emph{term symbols}, while symbols in $\FFl$ are \emph{theory symbols}.
A term in $\TT(\FFl,\VV)$ is a \emph{theory term}.
For a non-variable term $t = f(\seq t)$, we write $\root(t)$ to obtain the
top-most symbol $f$.
A \emph{position} $p$ is an integer sequence used to identify subterms,
and the subterm of $t$ at position $p$ is denoted $t|_p$.
We write $\Pos(t)$ for the set of positions in a term $t$, and given a set
of function symbols $\FF'$, $\Pos_{\FF'}(t)$ are those positions 
$p \in \Pos(t)$ such that $t|_p$ is rooted by a symbol in $\FF'$.
A \emph{substitution} $\sigma$ is a mapping from variables to terms with 
finite domain, and
$t\sigma$ denotes the application of $\sigma$ to a term $t$.

Theory terms $\TT(\FFl,\VV)$ have a fixed semantics:
we assume a mapping $\II$ that assigns to every sort $\iota$ occurring
in $\FFl$ a carrier set $\II(\iota)$. Moreover, we assume that for every
element $a \in \II(\iota)$ there is exactly one constant 
symbol $c_a \in \FFl$, called a \emph{value}.
The set of all value symbols is denoted $\Val$.
For instance, if the sort of integers occurs in $\FFl$ then 
$\Val \subseteq \FFl$ contains a value $c_i$ for every $i\in \mathbb Z$.

Moreover, we assume a fixed interpretation $\JJ$ that
assigns to every theory symbol $f \in \FFl$ a function $f_\JJ$ of appropriate sort,
and such that $(c_a)_\JJ = a$ for value symbols $c_a$, i.e., value symbols
are interpreted as the represented element.
The interpretation $\JJ$ naturally extends to
theory terms without variables by setting 
$\interp{f(\seq t)} = f_\JJ(\interp{t_1},\dots,\interp{t_n})$.
In particular, we assume a sort $\bool$ such that
$\II(\bool) = \{ \top, \bot \}$ with values
$\Val_\bool = \{ \true, \false \}$ such that $\true_{\JJ} = \top$, and
$\false_{\JJ} = \bot$. We also assume that $\FFl$ contains equality
symbols $\equals_\iota$ for every theory sort $\iota$, and a symbol $\wedge$ 
interpreted as logical conjunction.

Theory terms of sort $\bool$ are called \emph{constraints}, and a 
\emph{constrained term} is a pair $(t,\varphi)$ of a term $t$ and a constraint $\varphi$.
A substitution $\gamma$ is a \emph{valuation} if its range is a subset of
$\Val$.
A constraint $\varphi$ is \emph{valid}, denoted $\models \varphi$, 
if $\interp{\varphi\gamma} = \top$  for all valuations $\gamma$, and
\emph{satisfiable} if $\interp{\varphi\gamma} = \top$ for some
valuation $\gamma$. We write $\psi \models \varphi$ if all valuations
that satisfy $\psi$ also satisfy $\varphi$.

\paragraph{Logically Constrained Rewriting.}
A constrained rewrite rule is a triple $\ell \to r~[\varphi]$
where $\ell, r \in \TT(\FF,\VV)$, $\ell \not\in\VV$,
$\varphi$ is a constraint,
and $\root(\ell) \in \FFt$.
If $\varphi = \true$ then the constraint is omitted.
For a rule $\rho\colon\ell \to r~[\varphi]$ we use $lhs(\rho) = \ell$
and $rhs(\rho)=r$ to denote its left- and right-hand sides, 
respectively.
A set of constrained rewrite rules is called a \emph{logically constrained
term rewrite system} (\emph{LCTRS} for short).
For an LCTRS $\RR$, its \emph{defined} symbols $\FFtD$ are all 
root symbols of left-hand sides, that is, 
$\FFtD = \{\root(\ell) \mid \ell \to r~[\varphi] \in \RR\}$.
In the remainder we assume that LCTRSs are left-linear, that is, all variables occur at most once in the left-hand side $\ell$ of a rule
$\ell \to r~[\varphi]$.\footnote{Non-left-linear rules are rare in
practice; and moreover repeated 
occurrences of a variable $x$ in $\ell$ can be substituted by a fresh variable $x'$, 
adding $x \approx x'$ to $\varphi$. Though this implies that
$x$ can only be substituted by theory terms in rewrite sequences, for 
innermost evaluation this is not a limitation.}
An LCTRS $\RR$ is a \emph{transition system} if all rules in $\RR$ are of
the form
$f(\seq \ell) \to g(\seq [m]r)~[\varphi]$ such that $f,g \in \FFt$, all
$\ell_i \in \VV$, and all $r_j$ are in $\lterms$; if moreover the background
theory associated with $\FFl$ is the theory of integers then $\RR$ is an
\emph{integer transition system} (ITS).

The fixed rewrite system $\RR_\calc$ is the (infinite) set of rules
$f(\seq{\ell}) \to u$ such that $f \in \FFl \setminus \Val$, 
$\ell_i \in \Val$ for all $1\, {\leqslant}\, i\, {\leqslant}\, n$, and
$u \in \Val$ is the value symbol of $\interp{f(\seq{\ell})}$.
A rewrite step using $\RR_\calc$ is called a \emph{calculation step}
and denoted $\tocalc$.
A \emph{rule step} $s \to_{\rho}^\sigma t$ 
using a rule $\rho\colon \ell \to r\:[\varphi]$ and substitution $\sigma$ 
satisfies $s = C[\ell\sigma]$, $t = C[r\sigma]$, and $\sigma$ respects $\varphi$;
where a substitution $\sigma$ is said to \emph{respect} a constraint
$\varphi$ if $\varphi\sigma$ is valid and $\sigma(x) \in \Val$ for
all $x \in \Var(\varphi)$.
The substitution in the notation $\to_{\rho}^\sigma$ is mostly omitted, 
and a rule step simply denoted $\to_\rho$. 
For an LCTRS $\RR$, we denote the relation
${\tocalc} \cup {\{ \to_\rho \}_{\rho \in \RR}}$ by $\to_\RR$.
The above rewrite step is \emph{innermost}, denoted 
$\smash{s\ito_{\rho} t}$, if all proper
subterms of $\ell\sigma$ are in normal form with respect to $\to_\RR$.
Given binary relations $R$ and $S$, we write $R/S$ for $S^* \cdot R \cdot S^*$.
For LCTRSs $\RR$ and $\SS$ we abbreviate $\smash{\ito_{\RR}/\ito_\SS}$ by $\smash{\ito_{\RR/\SS}}$,
and $\smash{\ito_{\RR}/\tocalc}$ by $\smash{\ito_{\RR/\calc}}$.

\begin{example}[continued from \exaref{maxlength}]
\label{exa:maxlength2}
\label{ex:1}
The LCTRS $\RR_2$ indicated in \exaref{maxlength},
expressing the predicate \texttt{max\_length}/3, makes use of
the sorts  $\ints$, $\lists$ and $\bool$.
Furthermore, $\FFl$ consist of symbols $\cons$ and $\m{[]}$ for lists, $\cdot$, $+$, $-$, $\leqslant$,  and $\geqslant$ as well as values $n$ for all  $n \in \Ints$, 
with the usual interpretations on $\Ints$ and lists of integers.
Then $\RR$ admits the following rewrite steps:
\begin{align*}
\m{len}([\m 1, \m 2],\m 2)
\to \m{len}([\m 2],\m 2 - \m 1)
\tocalc \m{len}([\m 2],\m 1)
\to \m{len}([],\m 1 - \m 1)
\tocalc \m{len}([],\m 0)
\end{align*}
\end{example}
Note that in LCTRS rewriting, calculation steps like the subtractions in \exaref{maxlength2} are explicit in the $\tocalc$ relation,
in contrast to ITSs or related formalisms~\cite{MS:2018}, where
simplification is implicit. 
Moreover, innermost rewriting is a rather natural restriction for LCTRSs:
By the definition of a rule step using some rule $\rho$, variables in 
the constraint of $\rho$ need to be substituted by values. Hence 
non-innermost steps are only possible if nested redexes occur below 
unconstrained variables.
For instance, in a term $\m{f}(\m{f}(\m 2))$ only the inner $\m{f}$ 
call constitutes a redex for the rule $\m{f}(x) \to x~[x > \m 0]$.

\paragraph{Algebras.}
We assume mappings $|\cdot|_\iota: \II(\iota) \to \Nats$ for every sort 
$\iota$, playing the role of norms to measure size.
For instance, one might take
the absolute values for integers, the size function for arrays,
and the unsigned integer value for bit-vectors.
The subscript $\iota$ in $|t|_\iota$ is omitted if the sort of $t$ 
is clear from the context.
\medskip

We consider well-founded algebras $\AA$ over the natural numbers and the Booleans, with 
interpretation functions $f^\AA$ for all $f \in \FFt \cup \FFl$, cf.~\cite{BN98,TeReSe}.
By $t^\AA$ we denote the interpretation of a term $t$ based on 
$\AA$, and by $[\alpha]_\AA(t)$ the interpretation of $t$ based on 
$\AA$ and valuation $\alpha$.
In order to bound complexity, we use algebras that incorporate the given
complexity measures:
\begin{definition}
A \emph{measure interpretation} is given by
an algebra $\MM$ with carrier $\Nats$, and measures $|\cdot|_\iota$ for all 
sorts $\iota$. The interpretation $t^\MM$ of a term $t$ is $|t|_\iota$ if $t\in \VV$ has sort $\iota$, and $f^\MM(t_1^\MM, \dots, t_m^\MM)$ if $t=f(\seq[m] t)$. 
In addition, we demand that 
$f^\MM([t_1]_\JJ^\MM,\dots,[t_n]_\JJ^\MM) \geq [f(t_1,\dots,t_n)]_\JJ^\MM$ for all values
$\seq t$.
\end{definition}

In the following we suit interpretations (aka ranking functions) to LCTRSs.
The ternary relation $\cgt{\cdot}^\MM$ is defined as $s \cgt{\varphi}^\MM t$ if and only if
$[\alpha]_\MM(s) > [\alpha]_\MM(t)$ is
satisfied for all valuations
$\alpha$ that respect $\varphi$.
Similarly, $s \cge{\varphi}^\MM t$ if and only if
$[\alpha]_\MM(s) \geq [\alpha]_\MM(t)$ holds for all valuations $\alpha$ that 
respect $\varphi$.

\begin{definition}
We call an LCTRS $\RR$ \emph{weakly compatible} with a measure interpretation
$\MM$ if 
$\ell \cge{\varphi}^\MM r$ for all $\ell \to r~[\varphi] \in \RR$, and
\emph{strictly compatible}
if $\RR$ is weakly compatible and in addition
$\ell \cgt{\varphi}^\MM r$ for some $\ell \to r~[\varphi]\in \RR$.
\end{definition}
\begin{example}
\label{exa:algebra}
Consider the measure 
interpretation $\MM$ such that $\mergeSort_3^\MM(x,y,z) = y$, 
$\m{merge}^\MM(x,y,z) = x$, $x +^\MM y = x +_{\mathbb N} y$, 
$x -^\MM y = \max(x -_{\mathbb N} y,0)$, $\geq^\MM$ is $\geq^\Nats$, and
$v^\MM = \max(v,0)$ for all $v \in \mathbb Z$.
The LCTRS $\RR'$ consisting of the rules (2), (4), and (8)
from \exaref{mergesort} is strictly compatible with $\MM$,
since the rules (2) and (8) are weakly decreasing, while (4) is
strictly decreasing.
\end{example}

\section{Complexity Framework}
\label{sec:framework}

An LCTRS $\RR$ is \emph{terminating} if $\to_\RR$ is well-founded.
In applications like static analysis, termination of a program is often not
enough and more precise resource guarantees are needed. In this section
we propose suitable runtime complexity notions for LCTRSs.

Following common notions in complexity analysis~\cite{AM16},
the \emph{derivation height} of a term $t$ wrt.\ a binary relation $\to$ is 
defined as follows:
$
\dh(t_0,\to) \defsym \sup\, \{ k \mid \exists \,t_1, \dots,t_k.\ t_0 \to \dots \to t_k\}
$.
We assume that an LCTRS $\RR$ is associated with a unique 
\emph{initial state} $(t_0, \varphi_0)$ such that $\varphi_0$ is a 
constraint and
$t_0 = \init(\vec x)$ is the
\emph{initial term}, for a vector of \emph{input variables} $\vec x = (\seq x)$ and a function symbol $\init$ that
does not occur on any right-hand side.
The intention is that we consider only rewrite sequences starting
at $t_0\sigma$, such that $\sigma$ is a valuation that respects $\varphi_0$. 
Sometimes $s_0$ will be used as a shorthand for $(t_0, \varphi_0)$.

For $\vec u, \vec v \in \Nats^k$, let $\vec u \leq_k \vec v$ 
abbreviate $\bigwedge_{i=1}^k u_i \leq v_i$.
Given $\vec t = (\seq[k]t)$, $|\vec t|$ denotes $(|t_1|, \dots, |t_k|)$,
and $\vec t\sigma$ denotes $(t_1\sigma, \dots, t_k\sigma)$ for any substitution $\sigma$.
For a term $t$, we write $\Varvec(t)$ for a vector containing $\Var(t)$ 
in a fixed order.

\begin{definition}
For an LCTRS $\RR$ and a constrained term $(t, \varphi)$
such that $\vec x = \Varvec(t)$,
the (innermost) \emph{runtime complexity} 
$\rc_{\RR}^{(t, \varphi)} \colon \Nats^n \to \Nats \cup \{\omega\}$ is defined as
\[
  \rc_{\RR}^{(t, \varphi)}(\vec m) = \sup\:\{ \dh(t\sigma,\ito_{\RR/\calc}) \mid  \:|\vec x\sigma| \leq_n \vec m\text{ for some $\sigma$ that respects }\varphi\}.
\]
\end{definition}
Thus, the runtime complexity of an LCTRS is the maximal number of 
innermost \emph{rule} steps
in a rewrite sequence that starts with a size-bounded instance of the 
initial state $(t, \varphi)$; 
calculation steps are not counted.
This is common in cost analysis, it also corresponds to the 
runtime complexity of a program or ITS~\cite{BEFFG16}, where
the number of transitions are counted but not 
simplifications of expressions.

\smallskip
\emph{Dependency pairs} are commonly used in termination and complexity analysis
of rewrite systems.
For termination of LCTRSs they were already used in earlier work~\cite{K13}.
For complexity analysis, stronger notions were developed for standard
rewriting: \emph{dependency tuples} (DTs)~\cite{NEG13},
\emph{weak}~\cite{HM08},
and \emph{grouped} dependency pairs~\cite{AM16}.
Since we consider innermost rewriting, we can use an LCTRS
variant of dependency tuples.
To that end, for every defined symbol $f$ we consider a fresh symbol
$f^\sharp$, and for a term $t = f(\seq t)$ write 
$t^\sharp$ to denote $f^\sharp(\seq t)$. 
\begin{definition}
\label{def:dt}
Consider a rule $\rho\colon \ell \to r ~[\varphi]$ such that
$\PosD(r)$ is sorted as $\seq[k] p$ with respect to a fixed order on positions.
Then the \emph{dependency tuple} $\DT(\rho)$ of $\rho$ is
the constrained rule $\ell^\# \to \langle (r|_{p_1})^\#, \dots, (r|_{p_k})^\#\rangle_k~[\varphi]$.
For an LCTRS $\RR$, $\DT(\RR) = \bigcup_{\rho \in \RR} \DT(\rho)$.
\end{definition}

Here $\langle \dots\rangle_k$ is a fresh tuple symbol for every arity $k$
(but the subscript will be dropped for simplicity).
\begin{definition}[Dependency Graph]
Let $\RR$ be an LCTRS and $\DD\subseteq \DT(\RR)$.
The \emph{dependency graph} (DG) is the directed graph with node set 
$\DD$ and edges from $s^\# \to \langle \seq{t^\#}\rangle~[\varphi]$ to $u^\# \to v~[\psi]$
if there is some $t_i^\#$ such that $t_i^\#\sigma \to_\RR^* u^\#\tau$,
for some substitutions $\sigma$ and $\tau$ and some $i$, $1 \leq i \leq n$.
\end{definition}
The DG is not computable in general, but approximation techniques
are well-known~\cite{K13,NEG13,AMS16,Giesl:2017}.
For instance, the graph in \secref{mergesort} constitutes a dependency
graph approximation for the LCTRS from \exaref{mergesort}.
Following Noschiniski \etal~\cite{NEG13}, we assume particular interpretation 
functions for the tuple operators $\langle \dots\rangle$.
To this end, let a \emph{DT-measure interpretation} $\MM$ be a measure 
interpretation that interprets 
$\langle\seq[k]t\rangle^\MM = t_1 + \dots + t_k$, for all $k \geq 0$.
\smallskip

Let the set of \emph{bound expressions} $\UB$ be inductively
defined as follows: 
(i) $|x|_\iota \in \UB$ for $x\in\VV$ of sort $\iota$,
(ii) $\mathbb Z \subseteq \UB$ and $\omega \in \UB$,
(iii) if $p,q \in \UB$ then $p+q$, $pq$, and $max(p,q)$ are in $\UB$, and
(iv) if $p \in \UB$ and $k\in \Nats$ then $k^p$, $p/k$, and $\log_k(p)$ are in $\UB$. 
Given $p,q\in\UB$, we write $p \leq q$ if $[\alpha]_\Nats(p) \leq [\alpha]_\Nats(q)$ for all 
substitutions $\alpha\colon \VV \to \Nats$.
For a bound expression $p\in \UB$ and $\vec m\in \Nats^n$ we also write 
$p(\vec m)$ to denote the substituted bound expression $p[m_i/x_i]_{1\leq i\leq n}$,
assuming  $\vec x\in \VV^n$ are the variables in the initial term $t_0 = \init(\vec x)$.
\smallskip

A triple $P=((t,\varphi),\DD,\RR)$ of a constrained term $(t,\varphi)$, 
a set of DTs $\DD$, and an LCTRS $\RR$ is called a \emph{(complexity) problem}.
Following Brockschmidt \etal~\cite{BEFFG16}, we next define time and size bound approximations.

\begin{definition}
For a complexity problem $((t,\varphi),\DD,\RR)$ with $\vec x=\Varvec(t)$, 
a function $\RA \colon \DD \to \UB$ is a \emph{runtime approximation} if,
for all $\rho \in \DD$ and $\vec m \in \Nats^n$,
\[\RA(\rho)(\vec m) \geq \sup\:\{ \dh(t\sigma,\ito_{\{\rho\}/\DD\cup\RR}) \mid 
|\vec x\sigma| \leq_n \vec m\text{ and $\sigma$ respects }\varphi\}.\]
\end{definition}
In words, a runtime approximation $T(\rho)$ over-approximates how often a
DT $\rho \in \DD$ can be used in a rewrite sequence starting from the initial state, expressed in terms of the input variables.
For instance, consider \exaref{mergesort} and let $(1^\#)$ be the DT corresponding
to rule $(1)$. Then the function $\RA$ such that 
$\RA(1^{\#}) = 1$ and $\RA(\rho)(|x_0|,|y_0|,|z_0|) = |x_0|^2$ 
for all other DTs $\rho\in\DD$ 
is a valid (though not optimal) runtime approximation.

For a complexity problem $((t,\varphi),\DD,\RR)$,
the set of \emph{entry variables} $\TV$ is the set of all
tuples $\tvar{\rho}{y}$ such that 
$\rho \in \DD$ and $y \in \Var(lhs(\rho))$.
\begin{definition}
For a complexity problem $((t,\varphi),\DD,\RR)$ with $\vec x=\Varvec(t)$, 
a function $\SA \colon \TV \to \UB$ is a \emph{size approximation} if
\[
\SA\tvar{\rho}{y}(\vec m) \geq \sup\:\{ |y\tau| \mid \exists \sigma,\:u.\ 
t\sigma \ito_{\RR\cup \DD}^{*} \cdot \ito_{\rho}^{\tau} u
\text{, }|\vec x\sigma| \leq_n \vec m
\}
\]
for $\tvar{\rho}{y} \in \TV$ such that substitution
$\sigma$ respects $\varphi$,
and $\vec m \in \Nats^n$.
\end{definition}
A size approximation over-approximates how large a
variable in the left-hand side of a rule
in $\DD$ can get in a rewrite sequence from the initial state,
again expressed in terms of the input variables.
A tuple $(\RA, \SA)$ is a \emph{bound approximation} for 
a complexity problem $P$ if $\RA$ and $\SA$ are runtime and size
approximations for $P$.
We next define a complexity framework in the spirit of Avanzini and 
Moser~\cite{AM16}.
\begin{definition}
Given a complexity problem $P=(s_0,\DD,\RR)$, a \emph{(complexity) judgement} is a statement 
$\judge{P}{(\RA, \SA)}$, for functions $\RA:\DD \to \UB$ and
$\SA: \TV \to \UB$.
The judgement is \emph{valid} if $(\RA, \SA)$ is a bound
approximation for $P$.
A \emph{complexity processor} is an inference rule on complexity
judgements of the following form:
\[
\frac{\judge{P_1}{(\RA_1,\SA_1)}, \dots, \judge{P_k}{(\RA_k,\SA_k)}}
{\judge{P}{(\RA,\SA)}} \quad\mathsf{Proc}
\]
and it is \emph{sound}
if $\judge{P}{(\RA, \SA)}$ is valid whenever all $\judge{P_i}{(\RA_i, \SA_i)}$ 
are valid. 
\end{definition}
\noindent
For a problem $P = (s_0,\DD,\RR)$ with initial state $s_0 = (\init(\vec x), \varphi)$, 
a DT $\ell\to r~[\psi]\in \DD$ is \emph{initial} if $\root(\ell) = \init^\#$.
The \emph{initial processor} for $P$ is given by
\[
\frac{}
{\judge{P}{(\RA, \SA_\omega)}} \quad\mathsf{Initial}
\]
where $\RA(\rho) = 1$ if $\rho$ is initial and $\RA(\rho) = \omega$ otherwise; and
$\SA_\omega\tvar{\rho}{x} = \omega$ for all $\tvar{\rho}{x} \in \TV$.
Since $\init^\#$ does not occur on any right-hand side by assumption,
the processor \textsf{Initial} is sound. For instance, the DT
$\init^\#(x,y,z) \to \mergeSort^\#(x,y,z)$ originating from rule (1)
in \exaref{mergesort} is initial.
For a problem $P = (s_0,\DD,\RR)$ and an expression $\CA\in \UB$, we sometimes write
$\judge{P}{((\CA)_\Sigma, \SA)}$ to express that there is a 
runtime approximation $\RA$ such that $\judge{P}{(\RA, \SA)}$  and
$\CA = \sum_{\rho \in \DD} \RA(\rho)$.

The next result states that valid judgements bound the runtime complexity
of LCTRSs. It can be proven in a similar way as \cite[Theorem 6]{AM16}, using
the properties of dependency tuples for innermost rewriting.
\begin{theorem}
\label{thm:dt}
If an LCTRS $\RR$ with initial state $(t,\varphi)$ admits the valid judgement
$\judge{((t^\#,\varphi), \DT(\RR), \RR\cup\RR_\calc)}{(\RA,\SA)}$
then 
$\rc^{(t,\varphi)}_{\RR} \leq \sum_{\rho \in \DT(\RR)} T(\rho)$ holds.
\end{theorem}

\section{Processors}
\label{sec:procs}

This section presents processors that implement
the complexity framework from \secref{framework}, in particular showing how
the respective ITS techniques~\cite{BEFFG16} carry over.

\subsubsection*{Interpretation Processors.}
Compatible interpretations are a standard
tool in resource analysis, cf.~\cite{NEG13,AM16,BEFFG16}.
We first present a processor using a measure interpretation that orients \emph{all} rules and DTs (cf. \cite[Theorem 3.6]{BEFFG16}).
For $p\in \UB$, let $[p]$ denote the bound expression obtained from $p$ by replacing all
coefficients in $p$ by their absolute values (such that the resulting expression is weakly monotone).

\begin{lemma}
\label{lem:interpretation}
Let $P=((t_0,\varphi_0),\DD,\RR)$ and $\MM$ 
a $\DT$-measure interpretation with which $\RR$ is weakly, and $\DD$ is strictly compatible. Then the following processor is sound,
where $\RA'(\rho) = [(t_0)^\MM]$ for all $\rho\in \DD_{>}$,
and $\RA'(\rho) = \RA(\rho)$ otherwise:
\[
\frac{\judge{P}{(\RA, \SA)}}
{\judge{P}{(\RA', \SA)}} \quad\mathsf{Interpretation}
\]
\end{lemma}
For instance, for \exaref{mergesort} one can take the interpretation 
$\MM$ such that $\m{split}^\MM = 0$ and $f^\MM = 1$ for all 
other $f\in \FFt$, and symbols in $\FFl$ are interpreted as in \exaref{algebra}.
$\RR_1$ is strictly compatible since all rules are weakly and rule (5) is strictly decreasing. This justifies a runtime approximation setting
by $\RA(5^\#) = 1 = \init^\#(\vec x)^\MM$.

Next, we adapt \cite[Theorem 3.6]{BEFFG16} to our setting, by which
runtime bounds can be obtained using an interpretation 
that orients the given LCTRS \emph{partially}.
For a dependency graph $\REG$ and some $\DD' \subseteq \DD$, let
$\pre(\DD')$ be the set of all edges $(\rho_1, \rho_2)$ in $\REG$, such that $\rho_1 \in \DD \setminus \DD'$ and 
$\rho_2 \in \DD'$.
Moreover, for a DT $\rho$ with $\Varvec(lhs(\rho)) = (y_{1}, \dots, y_{k})$, let 
$\vec \SA_{\rho}$ denote $(\SA\tvar{\rho}{y_{1}}, \dots, \SA\tvar{\rho}{y_{k}})$.

\begin{lemma}
\label{lem:timebounds}
Suppose $P=\cprob$ is a complexity problem such that
$\DD' \subseteq \DD$ has no initial DTs,
$\RR$ is weakly, and $\DD'$ is strictly compatible with a
$\DT$-measure
interpretation $\MM$. Then the following processor is sound:
\[
\frac{\judge{P}{(\RA, \SA)}}
{\judge{P}{(\lambda\rho.
\left\{\begin{array}{ll}\sum_{(\gamma, \delta) \in \pre(\DD')} \RA(\gamma) \cdot
[lhs(\delta)^\MM](\vec \SA_{\delta}) &\text{ if }\rho \in \DD'_>\\
\RA(\rho)
&\text{ otherwise}\end{array}\right\}, \SA)}} \quad\mathsf{TimeBounds}
\]
where $\DD'_>$ is the set of rules $\ell \to r~[\varphi]$ in $\DD'$
such that $\ell \cgt{\varphi}^\MM r$.
\end{lemma}
\noindent
Next, we define a proceessor to compute size approximations.

\subsubsection*{Size Bounds.}
Size approximations were developed for ITSs and tend to be less precise for LCTRSs
due to nested terms. 
However, in many practical examples, 
a sufficient approximation is feasible. Next, we thus adapt the relevant notions to the LCTRS setting.
First, the \emph{local size approximation} overapproximates the size of entry
variables in terms of variable sizes in predecessor rules.
\begin{definition}
For $\delta,\rho \in \DD$ and $\tvar\rho{y} \in \TV$, let $\SA_{\delta \to \rho}\colon \VV \to \UB$ be a \emph{local size approximation} if
\[
\SA_{\delta \to \rho}(y)(\vec m) \geq \sup\:\{ |y\tau| \mid \exists t,\sigma
.\ \ell\sigma \to_{\delta}^{\sigma} \cdot \to_{\rho}^{\tau} t
\text{ and }\vec z\sigma\leq_n \vec m
\}
\]
where $\ell = lhs(\delta)$, $\vec z = \Varvec(\ell)$, and
$\sigma$ is a valuation.
\end{definition}
The intention is that for an entry variable $(\rho,y)$, such that $y$ occurs in the left-hand side of $\rho$, the expression $\SA_{\delta \to \rho}(y)$ upper-bounds $y$ in terms of the variables in $\delta$, for the case where
$\rho$ is applied after $\delta$. While such an expression is not 
always computable, it can often be over-approximated.
For instance, in \exaref{mergesort} a local size approximation
$S_{(9)\to(2)}(y)$ could be $(|x|+1)/2$ or $|x|$:
the subterm $\smash{\m m_3^\#(x,u,v)}$
on the right-hand side of (9) matches the left-hand side of $(2)$,
instantiating the variable $y$ by $u$,
and the side condition of (9) ensures $x+1 \geq 2u$.
We next define the \emph{entry variable graph} to track the dependence of
entry variables on each other.
For $f\in \UB$, let $\Var(f)$ be the 
set of all variables occurring in $f$.\footnote{For more precision one could restrict to \emph{active} variables, as done in~\cite{BEFFG16}.}

\begin{definition}
An \emph{entry variable graph} $\TVG$ for
$\cprob$ with DG $\REG$ has node set $\TV(\DD)$, 
and there is an edge from $\tvar{\delta}z$ to 
$\tvar{\rho}y$ labeled $\SA_{\delta \to \rho}(y)$ if
$\REG$ has an edge from $\delta$ to $\rho$ and 
$z \in \Var(\SA_{\delta \to \rho}(y))$.
\end{definition}
We illustrate the concept on our running example.
\begin{example}
\label{exa:mergesort2}
Consider again \exaref{mergesort}. We first apply \emph{chaining}, a standard technique
in termination an complexity analysis~\cite{FKS11,AMS16}, to compress the cycles 
$(9) - (3) - (9)$ and $(9) - (7) - (9)$ into single-step cycles, such that 
(9) is replaced by
\begin{footnotesize}
\begin{align*}
\mergeSort(x,y,z) &\to \langle \mergeSort_0(x,u,v),\mergeSort(u,u,v),\mergeSort(v,u,v),\mergeSort_3(x,u,v)\rangle \ [\psi] \\
\psi &= x \geq 2 \wedge u \geq 0 \wedge v \geq 0 \wedge x + 1 \geq 2u \wedge 2u \geq x \wedge x \geq 2v \wedge 2v + 1 \geq x.
\end{align*}
\end{footnotesize}
Then we obtain the following entry variable graph:
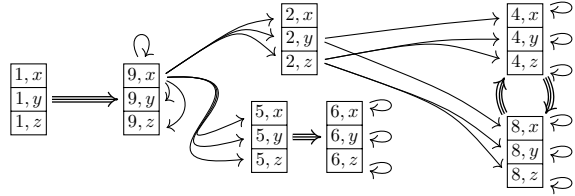
\begin{center}
\begin{tikzpicture}[align=center, xscale=2, node distance=15mm,] 
\tikzstyle{to}=[->,shorten <=1mm,shorten >=1mm]
\tikzstyle{tothree}=[->,shorten <=1mm,shorten >=1mm,\triple{yshift=.5mm}{yshift=-.5mm}{}]
\tikzstyle{dt}=[inner sep=2pt, scale=.7]
\tikzstyle{dts}=[draw, shape=rectangle split, rectangle split parts=3]
\tikzstyle{myloop}=[min distance=3mm,in=-20,out=20,looseness=5]
\node[dt,dts] (1) {(1)%
\nodepart{one}$1,x$\nodepart{two}$1,y$\nodepart{three}$1,z$};
\node[right of=1,dt,dts] (9) {%
\nodepart{one}$9,x$\nodepart{two}$9,y$\nodepart{three}$9,z$};
\node[dt,dts] (2) at (1.8,.8)  {%
\nodepart{one}$2,x$\nodepart{two}$2,y$\nodepart{three}$2,z$};
\node[dt,dts] (5) at (1.6,-.5) {%
\nodepart{one}$5,x$\nodepart{two}$5,y$\nodepart{three}$5,z$};
\node[dt,dts] (6) at (2.1,-.5) {%
\nodepart{one}$6,x$\nodepart{two}$6,y$\nodepart{three}$6,z$};
\node[right of=2,xshift=15mm,dt,dts] (4) {%
\nodepart{one}$4,x$\nodepart{two}$4,y$\nodepart{three}$4,z$};
\node[below of=4,dt,dts] (8) {%
\nodepart{one}$8,x$\nodepart{two}$8,y$\nodepart{three}$8,z$};
\triplearrow{arrows={-Implies}}{(1) -- (9)};
\draw[to] (9.90) to[loop above] (9.90);
\draw[to] (9.50) to[bend left=30] (9.-20);
\draw[to] (9.50) to[bend left=50] (9.-60);
\draw[to] (9.50) .. controls(1.4,1.) .. (2.130);
\draw[to] (9.50) .. controls(1.4,1.) .. (2.180);
\draw[to] (9.50) .. controls(1.4,1.) .. (2.230);
\draw[to] (9.50) .. controls(1.8,.4) and (.6,-1) .. (5.130);
\draw[to] (9.50) .. controls(1.8,.4) and (.55,-1).. (5.180);
\draw[to] (9.50) .. controls(1.8,.4) and (.5,-1) .. (5.230);
\draw[to] (2.0) .. controls(2.7,1) .. (4.130);
\draw[to] (2.-50) .. controls(2.7,.8) ..  (4.180);
\draw[to] (2.-50) .. controls(2.7,.8) ..  (4.230);
\draw[to] (2.0) .. controls(2.8,0.1) .. (8.130);
\draw[to] (2.-50) .. controls(2.8,-.1) ..  (8.180);
\draw[to] (2.-50) .. controls(2.8,-.1) ..  (8.230);
\triplearrow{arrows={-Implies}}{(8) to[bend left=15]  (4)};
\triplearrow{arrows={-Implies}}{(4) to[bend left=15]  (8)};
\triplearrow{arrows={-Implies}}{(5) -- (6)};
\draw[to] (6.0) to[loop right, min distance=3mm,in=-35,out=40,looseness=5] (6.0);
\draw[to] (6.60) to[loop right, min distance=3mm,in=-35,out=40,looseness=5] (6.60);
\draw[to] (6.-60) to[loop right, min distance=3mm,in=-35,out=40,looseness=5] (6.-60);
\draw[to] (4.0) to[loop right, min distance=3mm,in=-35,out=40,looseness=5] (4.0);
\draw[to] (4.60) to[loop right, min distance=3mm,in=-35,out=40,looseness=5] (4.60);
\draw[to] (4.-60) to[loop right, min distance=3mm,in=-35,out=40,looseness=5] (4.-60);
\draw[to] (8.0) to[loop right, min distance=3mm,in=-35,out=40,looseness=5] (8.0);
\draw[to] (8.60) to[loop right, min distance=3mm,in=-35,out=40,looseness=5] (8.60);
\draw[to] (8.-60) to[loop right, min distance=3mm,in=-35,out=40,looseness=5] (8.-60);
\end{tikzpicture}
\end{center}
where a triple arrow $a$\tikz{\node[inner sep=0pt](a){}; \node[inner sep=0pt,xshift=7mm](b){}; \triplearrow{arrows={-Implies}}{(a) -- (b)};}$b$ means that there are arrows from
$(a,x)$ to $(b,x)$, $(a,y)$ to $(b,y)$, and $(a,z)$ to $(b,z)$.
For all $(a,u)\in \TV$, all outgoing edges from $(a,u)$ can be labelled $|u|$, though
more precise approximations are possible.
\end{example}

\newcommand{\SAtriv}{\SA_{triv}}
\newcommand{\SAscc}{\SA_{scc}}

Next, we use the entry variable graph $\TVG$ to obtain size bound refinements,
following the approach of \cite{BEFFG16}.
To that end, we define two processors $\SAtriv$ and $\SAscc$ that refine bounds for
trivial and non-trivial SCCs in $\TVG$, respectively.
Here, an SCC is \emph{trivial} if it consists of a single node without
an edge to itself.

\begin{definition}
For size bounds $\SA$, we define $\SAtriv$ as follows:
  \begin{inparaenum}[(i)]
  \item $\SAtriv\tvar{\rho}{y} = |y|$ if  $\rho$ is initial;
  \item $\SAtriv\tvar{\rho}{y} = \max\{ \alpha(\vec \SA_\delta) \mid \tvar{\delta}{z} \to^\alpha\tvar{\rho}{y} \text{ in }\TVG \}$,
    if $(\rho,y)$ is not in any non-trivial SCC of $\TVG$;
  \item otherwise $\SAtriv\tvar{\rho}{y} = \SA\tvar{\rho}{y}$.
  \end{inparaenum}
\end{definition}

We distinguish three types of edges in $\TVG$, by partitioning their labels
into the three sets $E_=$, $E_+$, and $E_\times$: for an edge labelled $\alpha$,
\begin{inparaenum}[(i)]
\item $\alpha \in E_=$ if $\alpha = a_\alpha \in \Nats$ or
$\alpha = |x|$ for some $x\in \VV$;
\item $\alpha \in E_+$ if $|x| + a_\alpha \geq \alpha$ for some
$x\in \VV$ and $a_\alpha \in \Nats$;
\item $\alpha \in E_\times$ if 
$c  +\sum_{x \in \mathcal X} a_x|x| \geq \alpha$ for $c,a_x\in \Nats$
and $\mathcal X \subseteq \VV$.
\end{inparaenum}
For an SCC $C$ in $\TVG$, let $C_\alpha$ denote the set of 
edge labels $\alpha$ of edges in $C$.
For an entry variable graph $\TVG$, let $\pre(\rho,y)$
be the set of all direct predecessors of $(\rho, y)$ in $\TVG$.
\begin{definition}
\label{def:size bounds sccs}
Let $(\RA,\SA)$ be a bound approximation and
$C$ a non-trivial SCC in $\TVG$. 
Then $\SAscc$ is defined as
  \begin{inparaenum}[(i)]
  \item  if $C_\alpha \subseteq E_=$ then $\SAscc\tvar{\rho}{y} = max \{\alpha \mid \alpha \in C_\alpha\}$,
\item if $C_\alpha \subseteq E_+$ then
let $\alpha_{pre} = max\{ S(\rho',z)\mid \tvar{\rho'}{z} \in\pre\tvar{\rho}{y}\setminus C \}$ and
\[\SAscc\tvar{\rho}{y} = max(\{\alpha_{pre}\}\cup\{a_\alpha \mid \alpha \in C_\alpha\}) +
\sum_{\rho\in \DD}\RA(\rho) \cdot max\{a_\alpha \mid \alpha \in C \setminus E_=\}\]
\item
and $\SAscc\tvar{\rho}{y} = \SA\tvar{\rho}{y}$ otherwise,
\end{inparaenum}
for all $\rho \in C$ and $\tvar{\rho}{y}\in \TV$.
\end{definition}
Both $\SAtriv$ and $\SAscc$ are similar to the bounds developed 
in~\cite{BEFFG16}, though we omitted the case for $E_\times$ for reasons
of space. We obtain soundness by similar proofs.
\begin{lemma}
\label{lem:size bounds}
The following processors are sound:
\[
\frac{\parbox{45mm}{$\judge{(s_0, \DD, \RR)}{(\RA, \SA)}$}}
{\parbox{45mm}{$\judge{\cprob}{(\RA, \SAtriv)}$}}
\quad
\frac{\parbox{45mm}{$\judge{(s_0, \DD, \RR)}{(\RA, \SA)}$}}
{\parbox{45mm}{$\judge{\cprob}{(\RA, \SAscc)}$}} 
\quad
\mathsf{Size\ Bounds}
\]
\end{lemma}

\section{Processors for Splitting and Loop Summary}
\label{sec:newprocs}

In this section we present new processors to decompose
a problem into subproblems, as well as to analyse loops based on recurrence
relations.

\subsubsection*{Splitting.}
\label{sec:split}

We first consider a processor that allows to decompose a problem of a certain shape into
two subproblems. To that end, let a subgraph be \emph{forward closed} if it is 
closed under successors.

\begin{definition}
Consider a problem $P = \cprob$ whose DG $\REG$ exhibits subgraphs 
$G_0$ and $G_1$ with node sets $\DD_0$ and $\DD_1$, respectively,
such that $\DD = \DD_0 \uplus \DD_1$,
all initial DTs of $P$ are in $\DD_0$,
and $G_1$ is forward closed.
Then $(\DD_0, \DD_1)$ is a \emph{splitting} for $P$.
\end{definition}
\noindent
A splitting thus decomposes a problem
according to the scheme illustrated in \figref{subproblems:splitting}.
\begin{figure}
\begin{subfigure}[b]{0.4\textwidth}
\begin{tikzpicture}[scale=.9]
\tikzstyle{dt}=[scale=.8]
\filldraw[draw = gray!40, fill=gray!10] (0,0) circle (9mm);
\filldraw[draw = gray!40, fill=gray!10] (3,0) circle (9mm);
\node[dt] (chi1) at (60:7mm) {$\delta_1$};
\node[dt] (chi2) at (.7, 0) {$\delta_2$};
\node[dt] (chi3) at (300:7mm) {$\delta_m$};
\node[dt] (rho1) at ($(3,0) + (480:7mm)$) {$\gamma_1$};
\node[dt] (rho2) at (2.3,0) {$\gamma_2$};
\node[dt] (rho3) at ($(3,0) + (240:7mm)$) {$\gamma_m$};
\draw[->, bend left=10] (chi1) to (rho1);
\draw[->] (chi2)--(rho2);
\draw[->, bend right=10] (chi3) to (rho3);
\foreach \a in {120, 180, ..., 270}
 \node (a\a) at (\a:7mm) {$\cdot$};
\foreach \a in {300, 360, ..., 420}
 \node (b\a) at ($(3,0) + (\a:7mm)$) {$\cdot$};
\draw[->] (a180)--(a240);
\draw[->] (a180)--(chi2);
\draw[->] (a240)--(chi3);
\draw[->] (chi3)--(chi2);
\draw[->] (chi2)--(chi1);
\draw[->] (chi1)--(a120);
\draw[->] (a120)--(a180);
\draw[->] (rho2)--(rho3);
\draw[->] (rho2)--(b360);
\draw[->] (rho3)--(b300);
\draw[->] (b300)--(b360);
\draw[->] (b360)--(b420);
\draw[->] (b420)--(rho1);
\draw[->] (rho1)--(rho2);
\draw[->] ($(a180) + (-.6,0)$) -- (a180);
\node at (200:1.3cm){$\DD_0$};
\node at ($(3,0) + (340:1.3cm)$) {$\DD_1$};
\node[dt] at (1.5,-.4) {$\dots$};
\end{tikzpicture}
\caption{splitting}
\label{fig:subproblems:splitting}
\end{subfigure}
\hfill
\begin{subfigure}[b]{0.4\textwidth}
\begin{tikzpicture}[scale=.9]
\tikzstyle{dt}=[scale=.8]
\filldraw[draw = gray!40, fill=gray!10] (0,0) circle (3mm);
\filldraw[draw = gray!40, fill=gray!10] (3,0) circle (9mm);
\node[dt, inner sep=6pt] (delta) at (0,0) {$\delta$};
\node[dt] (gamma1) at ($(3,0) + (480:7mm)$) {$\gamma_1$};
\node[dt] (gamma2) at (2.3,0) {$\gamma_2$};
\node[dt] (gamma3) at ($(3,0) + (240:7mm)$) {$\gamma_m$};
\draw[->, bend left=10] (delta) to (gamma1);
\draw[->] (delta)--(gamma2);
\draw[->, bend right=10] (delta) to (gamma3);
\foreach \a in {300, 360, ..., 420}
 \node (b\a) at ($(3,0) + (\a:7mm)$) {$\cdot$};
\draw[->] (gamma2)--(gamma3);
\draw[->] (gamma2)--(b360);
\draw[->] (gamma3)--(b300);
\draw[->] (b300)--(b360);
\draw[->] (b360)--(b420);
\draw[->] (b420)--(gamma1);
\draw[->] (gamma1)--(gamma2);
\draw[->] ($(delta) + (-.6,0)$) -- (delta);
\node at ($(3,0) + (340:1.3cm)$) {$\DD'$};
\node[dt] at (1.5,-.3) {$\dots$};
\draw[->, loop below, min distance=4mm,in=250,out=230,looseness=4] (delta)
to (delta);
\draw[->, loop below, min distance=4mm,out=310,in=290,looseness=4] (delta) 
to (delta);
\node[scale=.6] at (0,-.7) {$1,\dots,p$};
\end{tikzpicture}
\caption{recurrence}
\label{fig:subproblems:recurrence}
\end{subfigure}
\caption{Problems of special shapes.\label{fig:subproblems}}
\end{figure}
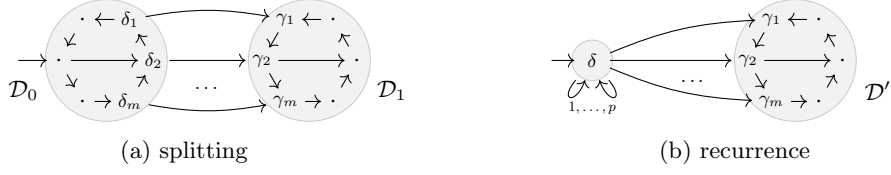
The idea is that we first analyse the subproblems $P_0$ and $P_1$ corresponding to 
$\DD_0$ and $\DD_1$ separately, considering as initial states for $P_1$ all possible entry 
points $\gamma_i$.
For DTs in $\DD_0$ their time bounds in $P_0$ constitute overall time bounds since $G_1$
is forward closed; on the other hand,
for every $\rho \in \DD_1$, we compute time bounds via each entry point $\gamma_i$, and obtain an overall time bound by taking the sum over all $\gamma_i$.
To that end, given $\gamma_i$, the time bound for $\rho$ in $P_1$ is applied to the size bound for
$\gamma_i$, and multiplied by the time bound for the respective $\delta_i$,
which upper-bounds the number of applications of $\delta_i$ followed by $\gamma_i$.
\begin{lemma}
\label{lem:split}
If $\cprob$ is a problem with splitting
$(\DD_0, \DD_1)$ such that
$\pre(\DD_1) = \{(\delta_i,\gamma_i) \mid 1\,{\leq}\,i\,{\leq}\,m\}$ and
$\gamma_i = (\ell_i \to r_i\ [\varphi_i])$,
the following processor is sound:
\[
\frac{\judge{P}{(\RA, \SA)} \quad \ 
\judge{(s_0, \DD_0, \RR)}{(\RA_0, \SA_0)} \quad \ 
\bigwedge_{i=1}^m\judge{((\ell_i, \varphi_i),\DD_1, \RR)}{(\RA_i, \SA_i)}}
{\judge{P}{(\lambda\rho.
\left\{\begin{array}{ll}T_0(\rho) &\text{ if }\rho \in \DD_0\\
\sum_{i=1}^mT_0(\delta_i) \cdot T_i(\rho)(\vec \SA_{\gamma_i})
&\text{ if }\rho \in \DD_1\end{array}\right\}, S)}}
\quad\mathsf{Split}
\]
\end{lemma}
Several improvements are conceivable, for instance the conditions of the
initial states $(\ell_i, \varphi_i)$ could be strengthened using reachability
analysis in the DG.

\subsubsection*{Summarising Self-Loops.}
\label{sec:loops}
We next propose a technique for the analysis of (sub)problems
whose DG is of the shape shown in \figref{subproblems:recurrence}.
For vectors $\vec a, \vec b$, let 
$\vec a >_k \vec b$ be a shorthand for the expression
$\vec a \geq_k \vec b \wedge (\bigvee_j a_j > b_{j})$.
\begin{definition}
Let $P = \cprobx[(f(\vec x), \varphi)]$ with DG $\REG$ such that
$\DD$ can be written as $\DD = \{\delta\} \uplus \DD'$, the graph
$\REG|_{\DD'}$ is forward-closed in $\REG$, and $\delta$ is of the form:
\begin{equation}
\label{eq:loop}
f(\vec x) 
\to 
\langle f(\vec r_{1}), \dots, f(\vec r_{p}), lhs(\gamma_1), \dots, lhs(\gamma_m)\rangle
\quad[\psi]
\end{equation}
for $\{\gamma_1, \dots, \gamma_m\} \subseteq \DD'$,
such that
$\vec x, \vec r_{i} \in \lterms^k$ and
$\varphi \wedge \psi \models |\vec x| >_k |\vec r_i|$
for all $1\,{\leq}\,i\,{\leq}\,p$.
If there is moreover some $\vec b \in (\Nats\,{\cup}\,\{-\infty\})^k$ such that
$\varphi \wedge \psi \models |\vec x| \geq_k \vec b$, then $P$ is \emph{cyclic} with \emph{termination condition} $\vec b$.
\end{definition}

\begin{lemma}
\label{lem:loops1}
Let $P = \cprob$ be a cyclic complexity problem with termination condition $\vec b$
and a DT $\delta$ of the form 
(\eqref{loop}), and let
$\gamma_i = (\ell_i \to r_i\ [\varphi_i])$,
for all $i$, $1\,{\leq}\,i\,{\leq}\,m$.
Then the following processor is sound:
\[
\frac{\judge{P}{(\RA, \SA)} \quad\bigwedge_{i=1}^m\judge{((\ell_i, \varphi_i), \DD', \RR)}{(\RA_i,\SA_i)}}
{\judge{\cprob}{(F(\vec x)_\Sigma, \SA)}} 
\quad
\mathsf{Recurrence}
\]
where $F$ is a solution to a recurrence
$f(\vec x) = f(\vec r_{1}) + \ldots + f(\vec r_{p}) + H(\vec x)$,
$f(\vec b) = 0$
for some $H(\vec x) \geq \sum_{\rho \in \DD'}\sum_{i=1}^m \RA_i(\rho)(\vec \SA_{\gamma_i})$.
\end{lemma}
This processor is key to analyse the main loop in our running example.
\begin{example}
Consider \exaref{mergesort} with chaining as applied in \exaref{mergesort2}.
For the subproblems
$P_1 = (\mergeSort_0^\#(x,u,v), \psi), \DD,\RR)$ and
$P_2 = (\mergeSort_3^\#(x,u,v), \psi), \DD,\RR)$
the judgements
$\judge{P_1}{((x+1)_\Sigma,\SA)}$
and
$\judge{P_2}{(u + v + 1)_\Sigma, \SA)}$
are valid, so we can set 
\newpage
\noindent
$H(x,u,v) = 2|x|+1 \geq x + u + v +1$ since $u,v \leq x/2$.
Thus, we solve the recurrence (\eqref{mergesort recursion}) given
in \secref{mergesort}. 
According to one of the cases of the Master Theorem, 
(\eqref{mergesort recursion}) has a solution in $\mathcal O(|x| \cdot \log(|x|))$
which is a complexity approximation according to \lemref{loops1}.
\end{example}

To simplify the presentation, we only considered cycles
formed by a single DT, as indicated in \figref{subproblems:recurrence}.
The result generalizes to
longer cycles, but chaining can often reduce these cases to the 
simpler situation discussed here.

\section{Evaluation}
\label{sec:implementation}

To evaluate the viability of the presented framework, we prototyped
our approach in the complexity analyser \TCT~\cite{AMS16}.

\paragraph{Implementation.}
We added a new module \texttt{tct-lctrs} to the \TCT tool suite,
below we call the resulting tool \tctlctrs.%
\footnote{The code is available from \url{https://github.com/bytekid/tct-lctrs}.}
It currently supports the theory of integers, as well as some operations on lists.
%
All processors described in this paper are implemented, using the modular processor
framework of \TCT. They are arranged in the following strategy,
where the loop indicates exhaustive repetition:
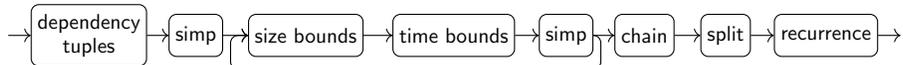
\begin{center}
\begin{tikzpicture}[node distance=15mm]
\tikzstyle{sstep}=[scale=.8, draw,rectangle, inner sep=3pt, minimum height=7mm, rounded corners=1mm,font=\sffamily]
\tikzstyle{to}=[->, rounded corners=1mm]
\node[sstep] (dt) {\begin{tabular}{@{}c@{}}dependency\\[-.2ex]tuples\end{tabular}};
\node[sstep] at ($(dt.east) + (.65,0)$) (simp1) {simp};
\node[sstep] at ($(simp1.east) + (1.1,0)$) (sb1) {size bounds};
\node[sstep] at ($(sb1.east) + (1.2,0)$) (tb1) {time bounds};
\node[sstep] at ($(tb1.east) + (.68,0)$) (simp2) {simp};
\node[sstep] at ($(simp2.east) + (.68,0)$) (chain) {chain};
\node[sstep] at ($(chain.east) + (.68,0)$) (split) {split};
\node[sstep] at ($(split.east) + (1,0)$) (loops) {recurrence};
\draw[to] ($(dt.west) + (-.3,0)$) -- (dt);
\draw[to] (dt) -- (simp1);
\draw[to] (simp1) -- (sb1);
\draw[to] (sb1) -- (tb1);
\draw[to] (tb1) -- (simp2);
\draw[to] (simp2) -- (chain);
\draw[to] (chain) -- (split);
\draw[to] (split) -- (loops);
\draw[to] (simp2) -- ($(simp2.east) + (.1,0)$) |- ($(simp1.east) + (.1,-.45)$) -- ($(simp1.east) + (.1,0)$) -- (sb1);
\draw[to] (loops) -- ($(loops.east) + (.3,0)$);
\end{tikzpicture}
\end{center}
We mention some implementation aspects that seem noteworthy.
\begin{itemize}
\item The \textsf{simp} processor combines some straightforward simplification processors: unsatisfiable paths, 
unreachable rules, and unused arguments are eliminated, and leaves in the DG obtain their time bound 
from their predecessors.
\item
Suitable algebras instantiating the interpretation and time bounds processors (\lemsref{interpretation}
{timebounds}) are searched for by means of an SMT encoding, as
done in the ITS module of \TCT previously using well-known techniques~\cite{PodelskiR04,BagnaraM13}.
\item
Before applying the recurrence processor, \TCT first applies chaining
to obtain loops that involve only a single DT (see Appendix B for details). 
\item In the recurrence processor (\lemref{loops1}),\TCT first attempts to solve 
subproblems corresponding to the functions $\seq[m]h$ separately, obtaining bound approximations
$(\RA_i, \SA_i)$ for all $i$, $1\leq i \leq m$ (see the notation of \lemref{loops1}).
Then, it is checked whether a function $H$ corresponding to one of the
known recursion patterns satisfies $H(\vec x) \geq \sum_i \sum_{\rho\in \DD'}T_i(\rho)$ using an SMT call.
\item 
The splitting processor (\lemref{split}) leaves a lot of choice to the 
implementation where to split. We currently use it to enable the loop processor,
which requires a very particular problem shape.
\end{itemize}
If a subroutine requires an SMT query, \TCT interfaces \textsf{Yices}~\cite{D14} and 
\textsf{Z3}~\cite{MB08}.

\paragraph{Experiments.}
We evaluated \tctlctrs on the ITS benchmarks considered by Brockschmidt \etal~\cite{BEFFG16}, using a timeout of 60 seconds.
\tabref{itss} compares our implementation with \koat~\cite{BEFFG16},
\cofloco~\cite{Montoya17,FM16}, the ITS version of 
\TCT~\cite{AMS16}, and \pubs~\cite{AAGP08}, giving 
\pagebreak\noindent
the number of problems for which a bound
was derived at all, the number of constant bounds, and the number of bounds that are at most linear, quadratic, and cubic, respectively.
\begin{table}
\begin{center}
\begin{tabular}{l|c|c|c|c|c}
 &\textsf{\tctlctrs} &\koat & \cofloco & \TCT-\textsf{ITS} & \pubs\\
\hline
solved problems & 359 & 404 & 347 & 309 & 285 \\ \hline
constant & 119 & 131 & 117 & 118 & 109 \\ \hline
$\leq \OO(n)$ & 282 & 298 & 270 & 250 & 240 \\ \hline
$\leq \OO(n^2)$ & 345 & 376 & 336 & 300 & 270 \\ \hline
$\leq \OO(n^3)$ & 356 & 383 & 345 & 306 & 278 
\end{tabular}
\end{center}
\caption{Comparison of tools on ITS benchmarks.\label{tab:itss}}
\vspace{-.5cm}
\end{table}
The new splitting and recurrence processors allow \tctlctrs to derive
sublinear bounds. This is the case for all problems where \pubs derives a (precise) logarithmic bound, such as
the examples \texttt{div\-ByTwo} and \texttt{direct\_n\_log\_n}. 
(\koat and \cofloco do not support sublinear bounds, and hence output linear
bounds for these examples.)
Moreover, we can precisely analyse subproblems produced by a divide-and-conquer approach like \texttt{divide\_and\_conquer}, where \TCT (as well as \koat) produces the tight linear bound, while \cofloco fails and \pubs gives an exponential bound.
Detailled results, including a complete table and \TCT output, are available on-line.\footnote{See %
\url{http://cl-informatik.uibk.ac.at/users/swinkler/lctrs_complexity/}}

We moreover tested \TCT on the set of logic programs collected by Mesnard and Neumerkl~\cite{MN01},%
\footnote{See~\url{http://www.complang.tuwien.ac.at/cti/bench/}.}
restricted to deterministic programs.
A list of solved problems is available on-line as well.

\section{Conclusion}
\label{sec:conclusion}

This paper presented the first complexity framework for LCTRSs. We conclude by relating to earlier work in the area, before indicating leads for future 
research.

\paragraph{Related work.}
In the last decades there has been significant progress in the area of \emph{fully automated} resource analysis,
showing that it can be both practicable and scalable, see e.g.~\cite{WEEHTWBFHMMPPSS:2008,AAGP08,G:2009,AGM:2013,SLH:2014,WG:2014,ADM:2015,FG:2017,HDW:2017,MS:2018}.
In the following, we indicate related work that directly influenced our framework, or employed similar methods.

Our framework differs from earlier work by Avanzini and Moser~\cite{AM16} in three
important respects: first, constraints over arbitrary background theories
are supported, second, complexity is not expressed in terms of the size of the initial 
term 
but in terms of measure functions,
and third, sublinear bounds can be derived.
While innermost rewriting is a rather
natural restriction for LCTRSs, \emph{call by need} strategies could be considered in the future for LCTRSs, too.

LCTRSs generalise ITSs, the complexity analysis of which is 
subject to a comprehensive line of research~\cite{BEFFG16,NEG13}.
Our approach gracefully extends the alternating time and size bound technique by
Brockschmidt \etal~\cite{BEFFG16}, as the ITS case is fully covered. In addition,
we can obtain sublinear bounds, and support further 
modularization. Moreover, LCTRSs offer native support for full recursion.

Sublinear bounds are beyond the scope of this earlier work, but can be
inferred by some other tools. 
Albert \etal~\cite{AAGP08} apply refinements to linear ranking
functions and support sufficient criteria for divide-and-conquer patterns. 
This allows the tool \pubs to recognize logarithmic and $\mathcal O(n\,\log(n))$
bounds for some problems.
Chatterjee \etal~\cite{ChatterjeeFG17} use synthesis ranking functions extended
by logarithmic and exponential terms, making use of an insightful
adaption of Farkas' and Handelman's lemmas. The
approach is able to handle examples such as \textsf{mergesort}. In contrast 
to our work this amounts to a whole-program analysis. Further, extensibility 
to a constraint formalism like LCTRS is unclear.
Wang \etal~\cite{WWC17} present an ML-like
language with type annotations, also using the Master Theorem to handle divide-and-conquer-like recurrences.
To estimate lower bounds for logic programs
based on divide-and-conquer,
Debray \etal~\cite{DLHL97} consider
non-deterministic recurrence relations and propose a technique to obtain
a closed-form bound for some cases.

\paragraph{Future work.}
We see exciting directions for future work both on a 
theoretical and an application level.
Various additional processors can be conceived
for our complexity framework, for instance forms of dependency pairs for
non-innermost rewriting~\cite{NEG13,HM08}, knowledge propagation and
narrowing~\cite{NEG13}.

Simplification systems as, for instance, employed in compiler toolchains (cf. 
\exaref{alive}) or SMT solvers constitute a highly relevant application 
domain, since these routines operate in performance-critical contexts.
In order to tackle such systems, techniques for \emph{derivational}
complexity of LCTRSs need to be developed.

On the application level, LCTRSs constitute a natural backend for complexity analysis of constraint
logic programs, since constraints can be natively expressed.
Our experiments with logic programs did not take backtracking into account,
but suitably adapting the transformational frameworks as established by
Giesl \etal~\cite{GSSEF:2012} to LCTRSs, 
this is not a showstopper: 
There the authors provide an automated complexity and termination analysis
  of full Prolog programs. In particular, the aforementioned restriction to deterministic
programs can be overcome.
We thus plan to support CLP as a frontend of our analysis, possibly
taking into account \emph{labelling strategies} that control the instantiation of query terms.
We furthermore plan to support C programs as a frontend.
C programs with integers, as considered in the Termination Competition\footnote{\url{http://termination-portal.org/}} can be 
expressed as ITSs. LCTRSs offer more flexibility and can support also
strings and floats, as the respective theories are supported by SMT
solvers. Just like for the case of CLP, this requires the development of
suitable complexity-reflecting transformations.
More experiments are planned to evaluate our method on
(constrained) logic programs~\cite{MN01} and problems from the software competition.\footnote{\url{https://sv-comp.sosy-lab.org/}}
\newpage
\bibliographystyle{abbrv}
\bibliography{references}

\appendix
\section{Proofs}

\begin{numberedlemma}{1}
Suppose $P=\cprobx$ is a complexity problem and $\MM$ 
a $\DT$-measure interpretation with which $\RR$ is weakly, and $\DD$ is strictly compatible. Then the following processor is sound:
\[
\frac{\judge{P}{(\RA, \SA)}}
{\judge{P}{(\RA', \SA)}} \quad\mathsf{Interpretation}
\]
where $\RA'(\rho) = (t_0^\#)^\MM$ for all $\rho\in \DD_{>}$,
and $\RA'(\rho) = \RA(\rho)$ otherwise.
\end{numberedlemma}
\begin{proof}
By assumption, $\RR$ has a unique initial state $(t_0, \varphi_0)$.
Let $\vec x=\Varvec(t_0)$, and $\rho \in \DD_>$.
We show that for
${\to_\TT} = {\to_{(\DD\cup\RR\cup\calc)\setminus \{\rho\}}}$,
\[
(t_0^\#)^\MM(\vec m) \geq \sup\:\{ k \mid \exists \sigma.\ 
t_0^\#\sigma \mathrel{(\to_\TT^* \cdot \to_\rho)^k} u\text{, }
|\vec x\sigma| \leq_n \vec m\text{ and $\sigma$ respects }\varphi_0\}
\]
for all $\vec m \in \Nats^n$.
Consider a substitution $\sigma$ such that $|\vec x\sigma| \leq_n \vec m$,
and a
(potentially infinite) rewrite sequence in which $k \in \Nats \cup \{\omega\}$
steps apply rule $\rho$:
\begin{equation}
\label{eq:MM sequence}
t_0^\#\sigma = v_0 \xrightarrow[\TT]{*} 
u_1 \xrightarrow[\rho]{} v_1 \xrightarrow[\TT]{*}
u_2 \xrightarrow[\rho]{} v_2 \xrightarrow[\TT]{*} \dots
\end{equation}
The goal is to verify $(t_0^\#)^\MM(\vec m) \geq k$. If $k=0$ then this is
obvious because the domain of $\MM$ is $\Nats$, so suppose $k>0$.
By compatibility, we have $v_{0}^\MM \geq u_1^\MM$, $u_{i}^\MM > v_i^\MM$, and
$v_i^\MM \geq u_{i+1}^\MM$, for all $i \geq 1$.
Thus $k$ must be finite, and there can be at most $(t_0^\#)^\MM(\vec m) $ 
$\rho$-steps in \eqref{MM sequence}.
\qed
\end{proof}

\begin{numberedlemma}
Suppose $P=\cprob$ is a complexity problem such that
$\DD' \subseteq \DD$ has no initial dependency tuples and
$\RR$ is weakly, and $\DD'$ is strictly compatible with a $\DT$-measure
interpretation $\MM$. Then the following processor is sound:
\[
\frac{\judge{P}{(\RA, \SA)}}
{\judge{P}{(\RA', \SA)}} \quad\mathsf{TimeBounds}
\]
such that 
$\RA'(\rho) = 
\sum_{(\gamma, \delta) \in \pre(\DD')} \RA(\gamma) \cdot
\ell_2^\MM(\vec \SA_{\delta})$ for all $\rho \in \DD'_>$,
where $\delta:\ell_2 \to r_2~[\varphi_2]$ and
$\vec \SA_{\delta}$ the vector
$\SA\tvar{\delta}{y_1}, \dots, \SA\tvar{\delta}{y_k}$
for $\seq[k]y$ the variables in $\ell_2$.
Otherwise, for all $\rho \in \DD \setminus \DD'_>$ set $\RA'(\rho) = \RA(\rho)$.
\end{numberedlemma}
\begin{proof}
It suffices to show that
\[
\RA'(\rho)(\vec m)
\geq \sup\:\{ k \mid 
t_0^\#\sigma \mathrel{(\to_\RR^* \cdot \to_\rho)^k} u\text{ and }
|\vec x\sigma| \leq_n \vec m\text{ and $\sigma$ respects }\varphi_0\}
\]
holds for all $\rho \in \DD'_{>}$ and $\vec m \in \Nats^n$.
Fix some $\rho \in \DD'_>$ and let 
$\to_\TT = \to_{(\DD\cup\RR\cup\calc)\setminus \{\rho\}}$.
Consider a substitution $\sigma$ such that $|\vec x\sigma| \leq_n \vec m$,
and a (potentially infinite) rewrite sequence
\begin{equation}
\label{eq:MM sequence2}
t_0^\#\sigma = v_0 \xrightarrow[\TT]{*} 
u_1 \xrightarrow[\rho]{} v_1 \xrightarrow[\TT]{*}
u_2 \xrightarrow[\rho]{} v_2 \xrightarrow[\TT]{*} \dots
\end{equation} with $k \in \Nats \cup \{\omega\}$ steps using $\rho$.
It has to be verified that the expression
$B := \sum_{(\gamma, \delta) \in \pre(\DD')} \RA(\gamma) \cdot
\ell_2^\MM(\vec \SA_{\delta})$ satisfies
$B(\vec m) \geq k$.
We can write sequence \eqref{MM sequence2} as
\begin{equation}
\label{eq:timebounds}
t_0^\#\sigma = v_0 \xrightarrow[\UU]{*} 
w_1' \xrightarrow[\UU]{}
w_1 \xrightarrow[\DD'\cup \RR\cup\calc]{*} v_1 \xrightarrow[\UU]{*}
w_2' \xrightarrow[\UU]{}
w_2 \xrightarrow[\DD'\cup\RR\cup\calc]{*} v_2 \xrightarrow[\UU]{*} \dots
\end{equation}
for $\UU = \RR\cup \DD \cup \calc \setminus \DD'$.
Every step $w_i' \to_\UU w_i$ before a term $w_i$ must
use a rule $\rho_i$ such that $(\rho_i, \delta_i) \in \pre(\DD')$,
for all $i>0$.
Fix some $(\gamma, \delta) \in \pre(\DD')$, and let $i$ be such that
$(\rho_i, \delta_i) = (\gamma, \delta)$.
Moreover, in the subsequence $w_i \to_{\DD'\cup\RR\cup\calc}^* v_i$ the value under the interpretation $\MM$ is weakly decreasing, and for every $\rho$-step
there is a strict decrease as $\rho \in \DD_>'$.
Hence the number of $\DD'_>$ steps (and in particular $\rho$-steps) in 
$w_i \to_{\DD'}^* v_i$ is bounded by $N:=lhs(\delta)^\MM(\vec \SA_{\delta})$.
The number of $\gamma$ steps followed by $\delta$ 
steps in \eqref{timebounds} is clearly bounded by $T(\gamma)$.
Summing up the product $N \cdot T(\gamma)$ over all
pairs $(\gamma,\delta) \in \pre(\DD')$ yields 
$B(\vec m) \geq k$.
\qed
\end{proof}

\setcounter{numberedlemma}{4}
\begin{numberedlemma}
If $\cprob$ is a problem with splitting
$(\DD_0, \DD_1)$ such that
$\pre(\DD_1) = \{(\delta_i,\gamma_i) \mid 1\,{\leq}\,i\,{\leq}\,m\}$ and
$\gamma_i = (\ell_i \to r_i\ [\varphi_i])$,
the following processor is sound:
\[
\frac{\judge{P}{(\RA, \SA)} \quad \ 
\judge{(s_0, \DD_0, \RR)}{(\RA_0, \SA_0)} \quad \ 
\bigwedge_{i=1}^m\judge{((\ell_i, \varphi_i),\DD_1, \RR)}{(\RA_i, \SA_i)}}
{\judge{P}{(\lambda\rho.
\left\{\begin{array}{ll}T_0(\rho) &\text{ if }\rho \in \DD_0\\
\sum_{i=1}^mT_0(\delta_i) \cdot T_i(\rho)(\vec \SA_{\gamma_i})
&\text{ if }\rho \in \DD_1\end{array}\right\}, S)}}
\quad\mathsf{Split}
\]
\end{numberedlemma}
\begin{proof}
Let $s_0 = (t_0, \varphi_0)$. We can represent an evaluation tree as
\begin{align}
\label{eq:splitting}
\{t_0\} &
\xrightarrow[\DD_0 \cup \RR]{*} \llangle u_{1}, \dots, u_{k}\rrangle
\xrightarrow[\delta_{j_1}, \dots, \delta_{j_k}]{*} \llangle v_{1}, \dots, v_{k'}\rrangle
\xrightarrow[\DD_1 \cup \RR]{*} \llangle w_{1}, \dots, w_{l}\rrangle
\end{align}
since the evaluation must start with $\DD_0$ steps, and there are no 
edges from $G_1$ to $G_0$.
(Here the notation $\llangle \dots\rrangle$ indicates a flattened
list of nested $\langle\dots\rangle$ terms.)
For any rule $\rho \in \DD_0$, by the shape of the evaluation tree (\ref{eq:splitting}),
$T_0(\rho)$ clearly serves as an overall time bound, since $\DD_0$ rules is never applied after 
$\DD_1$ rules.

So let $\rho \in \DD_1$. 
Note that for every $j$, $1\,{\leq}\,j\,{\leq}\,k$ such that
$u_j \to \langle v_{j_1}, \dots v_{j_q}\rangle$ and all $\hat v \in \{v_{j_1}, \dots, v_{j_q}\}$
there must be some $(\delta_i, \gamma_i) \in \pre(\DD_1)$ such that
the step uses $\delta_i \in \{\delta_1, \dots, \delta_m\}$ and $\hat v$ is an instance of $lhs(\gamma_i)$.
To obtain a time bound for $\rho$, we can thus estimate the number of applications of $\rho$ in a
subtree starting with an application of $\delta_i$ followed by $\gamma_i$, and take the sum
over all such pairs.
So let $(\delta_i,\gamma_i) \in \pre(\DD_1)$. To overapproximate the number of $\rho$-steps below 
$\hat v$, we consider $\RA_1(lhs(\gamma_i))$: this yields an expression
in $\Var(lhs(\gamma_i))$, and we obtain a respective expression
in $\Var(t_0)$ by applying $\RA_1(lhs(\gamma_i))$ to
$\vec \SA_{\gamma_i}$. In order to account for multiple
occurrences of the consecutive pair $(\delta_i,\gamma_i)$ in (\ref{eq:splitting}),
this expression is multiplied by $\RA_0(\delta_i)$.
\qed
\end{proof}

\begin{numberedlemma}
Let $P = \cprob$ be a cyclic problem with termination condition $\vec b$
and a DT $\delta$ of the form 
(\eqref{loop}), and let
$\gamma_i = (\ell_i \to r_i\ [\varphi_i])$,
for all $i$, $1\,{\leq}\,i\,{\leq}\,m$.
Then the following processor is sound:
\[
\frac{\judge{P}{(\RA, \SA)} \quad\judge{\bigwedge_{i=1}^m ((\ell_i, \varphi_i), \DD', \RR)}{(\RA_i,\SA_i)}}
{\judge{\cprob}{(F(\vec x), \SA)_\Sigma}} 
\quad
\mathsf{Recurrence}
\]
where $F$ is a solution to a recurrence
\begin{align}
\label{eq:recH}
f(\vec x) &= f(\vec r_{1}) + \ldots + f(\vec r_{p}) + H(\vec x) &
f(\vec b) &= 0
\end{align}
for some $H(\vec x) \geq \sum_{\rho \in \DD'}\sum_{i=1}^m \RA_i(\rho)(\vec \SA_{\gamma_i})$.
\end{numberedlemma}
\begin{proof}
As $\rho$ has no incoming edges in $G$ except from itself, $G$ is
of the form depicted in \figref{subproblems:recurrence}.
Let $\CA_i(\vec x) = \sum_{\rho \in \DD'} \RA_i(\rho)(\vec \SA_{\gamma_i})$
for all $i$, $1 \leq i \leq m$.
So an overapproximation of a solution to 
\begin{align}
\label{eq:recHH}
f(\vec x) &= f(\vec r_{1}) + \ldots + f(\vec r_{p}) + \sum_{i=1}^m \CA_i(\vec x) &
f(\vec b) &= 0
\end{align}
is a complexity approximation for $(s_0, \DD, \RR)$.
It remains to show that the overapproximation of $C(\vec x) :=\sum_{i=1}^p \CA_i(\vec x)$
by $H(\vec x)$ is sound.
To that end, let $F(\vec x)$ be a solution to \eqref{recH} and
$F'(\vec x)$ a solution to \eqref{recHH}. 
We show by induction on $\vec x$ that $F(\vec x) \geq F'(\vec x)$ holds for 
all $\vec x \geq_n \vec b$. By the assumption that $F$ and $F'$ solve \eqref{recH} and \eqref{recHH}, this clearly holds for 
the base case $\vec x = \vec b$.
For the step case, it is easy to verify that
$F(\vec x) \geq F'(\vec x)$ holds:
\begin{align*}
F(\vec x)
& =
F(\vec r_{1}) + \ldots + F'(\vec r_{p}) + H(\vec x) & 
\text{($F$ solves \eqref{recH})}\\
& \geq
F'(\vec r_{1}) + \ldots + F'(\vec r_{p}) + C(\vec x) & 
\text{(IH as $\vec r_i > \vec x$, and $H(\vec x) \geq C(\vec x)$)}\\
&= F'(\vec x)& 
\text{($F'$ solves \eqref{recHH})}
& \text{\qed}
\end{align*}
\end{proof}

\section{Further Processors}

In this section we outline how \emph{chaining} fits into our framework, a
standard technique from termination and 
complexity analysis.
Its adaptation is straightforward, hence we describe it only in
the appendix for the sake of completeness.

\subsubsection*{Chaining.}
Rule chaining is a commonly used complexity-reflecting transformation~\cite{FKS11,AMS16}
that can also be used for LCTRSs. 
Its aim is to combine  rules that are applied subsequently, which can simplify further analysis.

\begin{definition}
Let $P$ be a problem $\cprob$ with DG $G$ such that $\rho\colon t \to u~[\varphi]\in \DD$ has successors
$\delta_i\colon \ell_i \to r_i~[\psi_i]$ in $G$, for $1\,{\leq}\,i\,{\leq}\,k$.
Then $\rho$ and $\Delta = \{\seq[k]\delta\}$ are \emph{chainable} if 
there are positions $p_i$ and substitutions $\sigma_i$ such that
$u|_{p_i} = \ell_i\sigma_i$, and all $\psi_i\sigma_i$ are in $\lterms$.
\footnote{We assume that these are all possibilities to create a redex in $u$
wrt. rules in $\Delta$.}
Then
$\m{chain}(\rho,\delta_i) = t \to u[r_i\sigma]_{p_i}~[\varphi \wedge \psi_i\sigma_i]$, and $\m{chain}(\rho,\Delta)=\{\,\m{chain}(\rho,\delta_i)\,\}_{1 \leq i\leq k}$.
\end{definition}

\begin{lemma}
\label{lem:chaining}
For a complexity problem $\cprob$, the following processor is sound:
\[
\frac{\judge{(s_0, \DD\cup\m{chain}(\rho,\Delta) \setminus \{\rho\}, \RR)}{(\RA, \SA)}}
{\judge{\cprob}{(\RA', \SA')}}
\quad
\mathsf{Chaining}
\]
where $\rho$ and  $\Delta$ are chainable, the updated runtime approximation 
is
$\RA'(\delta_i) = \RA(\delta_i) + \RA(\m{chain}(\rho,\delta_i))$,
$\RA'(\rho) = \sum_i \RA(\m{chain}(\rho,\delta_i))$
and $\RA'(\gamma) =\RA(\gamma)$ otherwise; and 
$\SA'(\rho,y) =\max \{ \SA(\m{chain}(\rho,\delta_i))\}_i$,
$\SA'(\delta_i,y)$ is $\max(\SA(\delta_i,y), y\sigma_i)$ if $y\sigma_i \in \lterms$ and $\omega$ otherwise,
and $\SA'(\gamma,y) =\SA(\gamma,y)$ for all other $\gamma \in \DD$.
\end{lemma}

\begin{example}
\label{exa:mergesort chaining}
In two chaining steps, the cycles $(9) - (3) - (9)$ and 
$(9) - (7) - (9)$ in 
\exaref{mergesort} can be chained to single-step cycles, such that 
(9) is replaced by
\begin{footnotesize}
\begin{align*}
\mergeSort(x,y,z) &\to \langle \mergeSort_0(x,u,v),\mergeSort(u,u,v),\mergeSort(v,u,v),\mergeSort_3(x,u,v)\rangle \ [\psi] \\
\psi &= x \geq 2 \wedge u \geq 0 \wedge v \geq 0 \wedge x + 1 \geq 2u \wedge 2u \geq x \wedge x \geq 2v \wedge 2v + 1 \geq x
\end{align*}
\end{footnotesize}
\end{example}
%
%
%
%

\end{document}